\documentclass[onecolumn,nofootinbib]{revtex4}
\usepackage{amsmath}
\usepackage{amssymb}
\usepackage{amsthm}
\usepackage{amsfonts}
\usepackage{graphicx}
\usepackage[latin1]{inputenc}
\usepackage[T1]{fontenc}
\newtheorem{lemma}{Lemma}

\newtheorem{remark}{Remark}
\newtheorem{theorem}{Theorem}

\newcommand{\vp}{\varphi}
\newcommand{\tr}{{\rm Tr }}
\newcommand{\C}{\mathbb{C}}

\newcommand{\R}{\mathbb{R}}

\newcommand{\h}[1]{\mathcal{#1}}
\newcommand{\hi}{\mathcal{H}}

\newcommand{\cc}[1]{\overline{#1}}

\begin{document}
\title{Tunneling times with covariant measurements}
%\dedication{Dedicated to Pekka Lahti on the occasion of his 60th birthday.}
\author{J. Kiukas}
\email{jukka.kiukas@itp.uni-hannover.de}
\affiliation{Institut f\"ur Theoretische Physik, Leibniz Universität Hannover, Appelstrasse 2, 30167 Hannover, Germany.}
\author{A. Ruschhaupt}
\email{andreas.ruschhaupt@itp.uni-hannover.de}
\affiliation{Institut f\"ur Theoretische Physik, Leibniz Universität Hannover, Appelstrasse 2, 30167 Hannover, Germany.}
\author{R. F. Werner}
\email{reinhard.werner@itp.uni-hannover.de}
\affiliation{Institut f\"ur Theoretische Physik, Leibniz Universität Hannover, Appelstrasse 2, 30167 Hannover, Germany.}

\begin{abstract} We consider the time delay of massive, non-relativistic,
one-dimen\-sional particles due to a tunneling potential. In this
setting the well-known Hartman effect asserts that often the
sub-ensemble of particles going through the tunnel seems to cross the
tunnel region instantaneously. An obstacle to the utilization of this
effect for getting faster signals is the exponential damping by the
tunnel, so there seems to be a trade-off between speedup and intensity.
In this paper we prove that this trade-off is never in favor of
faster signals: the probability for a signal to reach its destination
before some deadline is always reduced by the tunnel, for arbitrary
incoming states, arbitrary positive and compactly supported tunnel
potentials, and arbitrary detectors. More specifically, we show this
for several different ways to define ``the same incoming state'' and
''the same detector'' when comparing the settings with and without
tunnel potential. The arrival time measurements are expressed in the
time-covariant approach, but we also allow the detection to be a
localization measurement at a later time.

\vspace{0.3cm}
\end{abstract}

 %\keywords{Arrival time\and tunneling\and Hartman effect}
\maketitle

\section{Introduction}

Questions related to the tunneling phenomenon have been actively
studied since the early days of quantum mechanics, and some of them
are still not resolved. In the simple case of a massive particle
moving in one dimension through a localized (tunnel) potential, the
question of the ''time spent in the tunnel'' is especially
interesting, and has given rise to extensive discussion (see e.g.
\cite{Hartman,Kijowski,Enders,Steinberg,Chiao,Nimtz} and the references
therein). Some difficulties in dealing with this problem are rooted
in the absence of a selfadjoint ''time operator'' (``Pauli's
Theorem'' \cite{Pauli}). Instead, one has to use more general
framework of positive operator measures (POMs)
\cite{Ludwig,Holevo,Werner,Busch}. For a survey on time in quantum
mechanics, see \cite{Muga}.

An old observation related to tunneling times is the so called
Hartman effect \cite{Hartman}, which states that the transmitted part
of a wave function appears to move faster through the tunnel than the
corresponding free state. More precisely, after a long rectangular
barrier and for a wave function of narrow momentum distribution, in
leading order the transmitted pulse appears at the end of the tunnel
{\it instantaneously}. Therefore, it has been suggested that this
effect means superluminal signal transport \cite{Enders,Nimtz}.
However, all this is only true for the shape of the wave function
disregarding normalization. But obviously, especially for long
tunnels, for which the gain in speed would only be noticeable, the
transmission probability is exponentially small. So in any attempt to
utilize the Hartman effect for a faster signal transmission, we would
have to analyze the trade-off between transmission probability and
transmission speed.

The main result of this paper is that this trade-off is always
trivial: when damping is taken into account, transmission through a
tunnel will always slow down the signal. Figures 1 and 2 show a sketch of the
result of an arrival time measurement in a possible situation: The
arrival time probability density for the transmitted particles peaks
earlier than for the free particles. The density may even become
larger at some times. But if we look at the integrated density, i.e.,
the probability for the particles to arrive before a given deadline
$t$, plotted here as a function of $t$, then the free particles win.

\begin{figure}
\includegraphics[width=5cm]{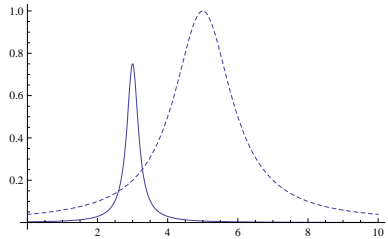}
\caption{Arrival probability density for free (dashed line) and transmitted (solid line)
particles}
\end{figure}
\begin{figure}
\includegraphics[width=5cm]{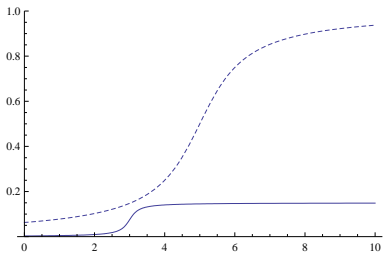}
\caption{Probability for arrival up to time t for free (dashed line) and transmitted (solid line)
particles}
\end{figure}

This is true in remarkable generality: for any incoming state, any
tunnel potential, and any detector. It is even true for several, in
general inequivalent approaches to formalizing the rules of this
race. Indeed we need to choose a precise notion of arrival detection,
of the equality of initial wave functions, and of the equality of
detectors for tunneling and free dynamics. The equality of initial
states is a non-trivial issue, because the two states are subject to
different dynamical evolutions. So at least we need to fix a
reference time. For this we have two choices, namely either a fixed
time (set to zero by convention), or $t=-\infty$, i.e., asymptotic
equality of the incoming states in the sense of scattering theory. In
this case we fix the direction from which the particles are coming by
choosing input states with positive momentum.

On the detection side, a natural choice is to describe the detectors
as covariant arrival time observables \cite{Holevo,Werner}. Again this
raises the issue of how to compare the two cases, because the
covariance condition explicitly depends on the time evolution, and an
observable can be covariant with respect to only one of them. Again
scattering theory helps, by defining a bijective correspondence
between the respective sets of covariant observables: we identify
observables, which give the same probability distributions on states
coinciding for $t\to+\infty$. This identification is also natural for
including finer descriptions of the detection process. For example,
we could modify both time evolutions by including an imaginary
``optical'' potential, resulting in contraction semigroups rather
than unitary groups. The loss of normalization is then interpreted as
arrival probability. Even more realistically, we could model the
detector by a system in a bound state, interacting with the particle
through a potential, and getting ionized in the detection process.
The ionization time (as measured by a covariant observable on the
escaping electron) is then taken as the detection time. It will be
shown elsewhere how these ideas lead to special cases
of covariant measurements.

In this paper we also look at another way of setting up the finish line
for the race: For particles traveling in the positive $x$-direction,
any time $t$, and any position $a$ behind (=to the right of) the
tunnel, we can replace the event ''the particle has arrived at point
$a$ before time $t$'' by ``the particle is located in the half axis
$[a,\infty)$ at time $t$''. The latter statement requires only a
position measurement, and hence does not require the theory of
arrival times. Although the two statements are not equivalent, and
correspond to different effect operators (the first probability is,
by definition, increasing as a function of $t$, but the second is not
in general increasing), they will be qualitatively similar, and equal
in the classical limit. Note that both the arrival times and the
localization observables (see e.g. \cite{Ali,Busch}) are represented
by positive operator measures, constrained by a covariance condition.
In order to emphasize the analogy with arrival times, for which a
projective measurement does not exist, we also allow localization
observables, which are general positive operator measures
(``POM'')\footnote{More commonly called ''POVMs'' for positive
operator {\it valued} measures}, rather than just the standard
position observable.

To summarize, we are looking at the following three approaches to our
problem:\\
\vskip12pt

\begin{tabular}{c|l|c|l}
&approach&initial reference time&detection\\\hline
I&covariant time&$-\infty$&time-covariant\\
II&localization& $-\infty$&localization\\
III&time-zero  & $0$&localization
\end{tabular}
\vskip12pt

The combination of the time-zero initialization with time-covariant
measurements is conspicuously absent from this table, because the
identification of detectors under different dynamics requires
asymptotic scattering theory, which we wanted to avoid in this
setting.

There would be more possibilities for the detection, with more
realistic detector descriptions, and some of these are now under
investigation. Approaches I and II have been discussed in a recent
paper by Werner and Ruschhaupt (to be published, see also a
conference report by Werner at the 40th Toru\'n Symposium, June 2008).
Here we have added the ''time zero approach'', as well as formulated
the treatment in a way that brings the use of positive operator
measures and covariance to the front.

\section{Preliminaries and notations}

As mentioned in the introduction, we have three distinct approaches, each formulated in terms of
covariant positive operator measures. Because of the covariance property (which we will precisely
define shortly below), the mathematical structure of the problem will be the same in each case, once the
observable is suitably transformed. More specifically, time observables are
defined via time translations and localization observables via space translations; in both cases, the
structure is the same in the Hilbert space where the translation generator acts multiplicatively.
Accordingly, we will transform into the ''energy representation'' $L^2([0,\infty),dE)$ for the arrival time
case, and into the ''momentum space'' $L^2(\R,dk)$ for the other two cases. Notationally, this is
conveniently implemented by defining the relevant basic operators (multiplication in particular) in the
generic $L^2(\R)$, and this is done in the following.

\subsection{Basic notations}

For any Hilbert space $\hi$, we let $\h L(\hi)$ denote the set of bounded operators on $\hi$.
Let $M$ be the multiplication operator acting in $L^2(\R)$ as $(M\vp)(x)= x\vp(x)$, on its domain of
selfadjointness. Let $D$ be the differential operator $i\frac{d}{dx}$, likewise in $L^2(\R)$.
These operators will be used in different forms: in the ''position representation'' $L^2(\R, dx)$, we will put
$Q:=M$ and $P:=-D$; they are the standard \emph{position and momentum operators}. In the
''momentum representation'' $L^2(\R, dk)$, the multiplication operator $M$ represents momentum, and
in the ''energy representation'' $L^2([0,\infty),dE)$ it acts as ''multiplication by energy''.

Let $F\in \h L(L^2(\R))$ be the Fourier-Plancherel operator, i.e. the unitary operator with
$$
(F\vp)(y) := \frac{1}{\sqrt{2\pi}}\int_\R e^{-iyx} \vp(x)\, dx, \ \ \vp\in L^2(\R)\cap L^1(\R).
$$
The operators $M$ and $D$ are well-known to be connected via the operator equalities
\begin{equation}\label{fourierconnection}
D= FMF^*=-F^*MF.
\end{equation}
(See e.g. \cite[pp. 106, 112]{Akhiezer}; it will be crucial to get the signs correctly.)
We will also denote $\hat{\psi}=F\psi$, and $\check{\psi} = F^*\psi$, for $\psi\in L^2(\R)$. For any
operator $A$ in $L^2(\R)$ (bounded or not), we denote $\hat{A} :=FAF^*$, and $\check{A} = F^*AF$.

In the physical context, we use the Fourier operator $F$ in the natural way as $F:L^2(\R, dx)\to L^2(\R, dk)$, so that the notations $\hat{\psi}$, $\check{\psi}$ are as usual. Also the meaning of $\hat{A}$ and $\check{A}$ should be clear: if $A$ is an operator in the position space, then $\hat{A}$ is the corresponding operator in the momentum space, and if $A$ acts in the momentum space, then $\check{A}$ is how it acts in the position space.  In particular, $\hat{Q}=D$ acts as a differential operator in $L^2(\R, dk)$, while $\hat{P}=M$ is the multiplication.

For any Borel function
$f:\R\to \C$, the operator $f(M)$, as defined via the spectral calculus, is simply the multiplication by $f$,
on its domain $\left\{\vp\in L^2(\R)\mid \int |f(x)\vp(x)|^2\, dx <\infty\right\}$.
We let $P_+=\chi_{[0,\infty)}(M)$ and $P_-=\chi_{(-\infty,0]}(M)$ (where $\chi_X$ is the characteristic
function of a set $X\subset \R$). Then $P_\pm\in \h L(L^2(\R))$ are projections, and we do the obvious
identifications $P_+L^2(\R)=L^2([0,\infty))$, $P_-L^2(\R)=L^2((-\infty, 0])$. Using the above defined
notation, we have $\check{P}_{\pm}=F^*P_{\pm}F$. Since $F^2$ is the parity operator, we
also have $\check{P}_{\mp} = I-\check{P}_{\pm}= FP_{\pm}F^*=\hat{P}_{\pm}$.

The reason for introducing these projections is that we will frequently need the subspaces of positive
and negative momenta. In $L^2(\R, dk)$, these are just $P_+L^2(\R, dk) = L^2([0,\infty),dk)$ and $P_-
L^2(\R, dk) = L^2((-\infty,0],dk)$, while in the position space $L^2(\R, dx)$, they are the images of the
projections $\check{P}_\pm$.

\subsection{Covariant observables}

Each of the three approaches is formulated in terms of positive
operator measures, defined on the Borel $\sigma$-algebra $\h B(\R)$
of the real line. We proceed to define this concept.

Let $\hi$ be a Hilbert space. A set function $E:\h B(\R)\to \h L(\hi)$ is said to be a \emph{positive operator
measure} (POM) if $E$ is strongly (or, equivalently, weakly) $\sigma$-additive, and $0\leq E(B)\leq I$ for
all $B\in \h B(\R)$. For any pair of vectors $\vp,\psi\in \hi$, and a POM $E:\h B(\R)\to \h L(\hi)$ we can associate the
complex measure $B\mapsto E_{\psi,\vp}(B) := \langle \psi|E(B)\vp\rangle$.

A POM will also be called \emph{observable}, when a quantum system is associated with the Hilbert
space $\hi$. The physical meaning of this is imported by postulating that for any \emph{state} operator $
\rho$ (i.e. a positive operator of trace one), the number
$\tr[\rho E(B)]$ is the probability that the measurement of $E$ yields a value from the Borel set $B$, given that
the system is prepared into the state $\rho$.
Note that here we do \emph{not} require an observable to be normalized in the sense that $E(\R)=I$.
The positive operator $I-E(\R)$ is simply interpreted as corresponding to the event of ''no detection''.

We will need two kinds of observables, \emph{arrival time} and \emph{localization} observables.
%We willdefine these observables in the position space $L^2(\R, dx)$, this obviously being, in the physical pointof view, the natural arena for setting up the tunneling problem.
In the first case the problem of \emph{time} in quantum mechanics is obviously involved. Without delving
into the long history of this question (see the references given in the introduction), we recall that the use
of POMs is forced by the fact that there is no selfadjoint operator giving eligible ''time'' probability
distributions.

A \emph{time observable} associated with the free evolution $H_0:=P^2$ is a POM $E:\h B(\R)\to \h L
(L^2(\R, dx))$ satisfying the \emph{covariance} condition
\begin{equation}\label{timeobs}
e^{itH_0} E(B)e^{-itH_0} = E(B+t), \ \ \text{for all } B\in \h B(\R),\, t\in \R.
\end{equation}
This encodes the minimal requirement that the measurement of $E $
performed at time $t>0$ gives a result from the range $[t_1,t_2]$
with the same probability as the measurement of $E$ at $t=0$ gives a
result from the shifted range $[t_1+t,t_2+t]$.

For an \emph{arrival time} observable, we additionally require that $E(\R)$ is the projection onto the
subspace of positive momenta, i.e. $E(\R)=\check{P}_+$. This is simply because the ''arrivals'' are
supposed to be coming only from the left, so the negative momentum part is not detected (see \cite
{Werner} for a more general formulation of screen observables.)
%This simplifies the setting considerably, since only the ''transmitted part'' of the scattered state has to be taken into account.

In the second (and third) approach, we need \emph{localization} observables. The standard localization
observable is given by the spectral measure $E^Q$ of the position operator $Q$. However, in order to emphasize the mathematical similarity of the approaches, we consider general
localization observables, i.e. POMs $E:\h B(\R)\to \h L(L^2(\R,dx))$, satisfying \emph{translation
covariance} and \emph{velocity boost invariance}:
\begin{align}\label{locobs}
e^{-itP} E(B)e^{itP} &= E(B+t), & e^{itQ} E(B)e^{-itQ} &= E(B), \ \ \text{ for all }  t\in \R.
\end{align}
Such observables have been studied in the context of approximate (or imprecise) position
measurements (see e.g. \cite{Davies,Ali,Busch,Carmeli}). In particular, they are all known to be of the
form $E=\mu * E^Q$, where $\mu:\h B(\R)\to \h L(\hi)$ is a probability measure, and the convolution is
defined in terms of the associated complex measures.

\subsection{The tunnel potential and scattering in one dimension}

Having defined the covariance concepts, we move on consider the
tunnel potential. Quite naturally, the essential quantity will turn
out to be the \emph{transmission amplitude} associated with the
scattering from the potential. The relevant information will be given
in Theorem \ref{transmission} below.

Let $H_0:=P^2$ be the free Hamiltonian (in the position representation). For our purposes, a \emph
{tunnel potential} is a (measurable) function $V:\R\to \R$ such that
\begin{itemize}
\item[(i)] $V$ is compactly supported and bounded;
\item[(ii)] The Hamiltonian $H=H_0+V$ has no eigenvalues (e.g. $V$ is positive).
\end{itemize}
Condition (i) assures that the tunnel is strictly localized in some interval $(x_0,x_1)$, and does not form
an impenetrable barrier. The second condition means that it actually acts as a barrier rather than e.g. a well. In order to not to exclude the square barriers typically used in the context of tunneling, we have not required continuity for the potential.

Next we need to recall some basic facts of scattering theory in one dimension. Under the above conditions defining the tunnel potential, the Hamiltonian $H = H_0+V$ is selfadjoint, with purely absolutely continuous spectrum $[0,\infty)$.  In particular, \emph{there are no bound states}. The \emph{wave operators}
$$
\Omega_{\pm} = s-\lim_{t\rightarrow\pm \infty} e^{itH}e^{-itH_0}
$$
exist and \emph{asymptotic completeness} holds, i.e. $\Omega_{+}(\hi)=\Omega_-(\hi)=L^2(\R,dx)$. The
operators $\Omega_{\pm}$ are unitary, and the unitary operator $S=\Omega_-^*\Omega$, which
connects incoming and outgoing asymptotics, is called the \emph{scattering operator}.
%The interpretation of these operator is the following: given a vector state
%$\psi$, evolving according to $e^{-itH}\psi$, the vector $\psi_0=\Omega_-^*\psi$ has the property
%that $e^{-itH_0}\vp_0$ and $e^{-itH}\vp$ are asymptotically equal at $t\rightarrow -\infty$. Similarly,
%the ''outgoing'' asymptotic state is given by $\psi_0'=\Omega_+*\psi$, so that $S\psi_0 = \psi_0'$
%connects the asymptotics, and allows for the calculation of transition probabilities.
We have the following intertwining relations.
\begin{align}
e^{-it H}\Omega_{\pm} &= \Omega_{\pm}e^{-it H_0}, &  t \in \R;\label{intertwining}\\
e^{-itH_0}S &= S e^{-it H_0}, & t \in \R.\label{scommutation}
\end{align}
The last equality implies that $S$ commutes with $H_0 = P^2$, and
%This means that it commutes with the spectral projections of momentum $P$ of the form $E^P(X\cup -X)$, where $X\subset [0,\infty)$ is a Borel set.
this gives rise to a decomposition of $S$: letting $\mathsf{P}$ denote the parity operator, each of the four operators
$\check{P}_+S\check{P}_+$, $\check{P}_+S\check{P}_-\mathsf{P}$, $\mathsf{P}\check{P}_-S\check{P}_+$, and $\mathsf{P}\check{P}_-S\check{P}_-\mathsf{P}$, acts on $\check{P}_+L^2(\R,dx)$, and commutes with the momentum $P$. The corresponding
momentum space operators thus act multiplicatively on $L^2([0,\infty), dk)$; we will denote them by
\begin{align*}
T_r(\hat{P}) &:= \mathsf{P}P_-\hat{S}P_-\mathsf{P}, & R_l(\hat{P}) &:= \mathsf{P}P_-\hat{S}P_+\mathsf{P}\\
R_r(\hat{P}) &:= P_+\hat{S}P_-\mathsf{P}, & T_l(\hat{P}) &:= P_+\hat{S}P_+.
\end{align*}
(Recall the notation: in $L^2(\R, dk)$ the momentum operator is $\hat{P}=M$.) The four functions $T_l, T_r, R_l, R_r:[0,\infty)\to \C$ thus  defined are measurable, and (essentially)
bounded by one.
The functions $T_l$ and $T_r$ are called the \emph{coefficients of transmission}, while $R_l$ and
$R_r$ are the \emph{coefficients of reflection}. By denoting
$$
\mathsf{S}(k) :=\begin{pmatrix} T_r(k) & R_l(k) \\ R_r(k) & T_l(k) \end{pmatrix}, \ \ k\geq 0,
$$
one gets an explicit form for the action of $S$ in the momentum space:
\begin{equation}\label{sconnection}
\begin{pmatrix} \hat{S}\hat{\psi}(-k)\\ \hat{S}\hat{\psi}(k)\end{pmatrix} = \mathsf{S}(k)\begin{pmatrix} \hat{\psi}(-k)\\ \hat{\psi}(k)\end{pmatrix}, \ \ \ k\geq 0, \ \ \hat{\psi}\in L^2(\R, dk).
\end{equation}
The $k$-dependent matrix $\mathsf{S}(k)$, $k\geq 0$, is called the \emph{scattering matrix} for $H$.
It mixes the positive and negative momentum components of the ''initial'' asymptotically free state to
produce the corresponding ''final'' free state, having ''transmitted'' and ''reflected'' parts.

The structure of transmission and reflection coefficients is investigated via the stationary scattering
theory: there exists, for each $k\in \R$, two solutions $\phi_1(x,k)$ and $\phi_2(x,k)$ of the differential
equation
\begin{equation}\label{schrodinger1}
-\frac{d^2}{dx^2}\psi(x) +V(x)\psi(x) = k^2\psi(x),
\end{equation}
analytically depending on $k$, and satisfying
\begin{eqnarray}
\phi_1(x,k)&=&\begin{cases} R_l(k)e^{-ikx}+e^{ikx}, & x\leq x_0\\ T_l(k) e^{ikx}, & x\geq x_1\end{cases}\label{asymptotic1}\\
\phi_2(x,k)&=&\begin{cases} T_r(k)e^{-ikx}, & x\leq x_0 \\ R_r(k)e^{ikx}+e^{-ikx}, & x\geq x_1\end{cases},\label{asymptotic2}
\end{eqnarray}
where $x_0,x_1\in \R$ are any two points such that the support of $V$ is included in $(x_0,x_1)$.
%\begin{eqnarray*}
%f_1(x,k)&\sim&\begin{cases} T_l(k)^{-1}(R_l(k)e^{-ikx}+e^{ikx}), & x\rightarrow -\infty\\ e^{ikx}, & x\rightarrow \infty\end{cases}\\
%f_2(x,k)&\sim&\begin{cases} e^{-ikx}, & x\rightarrow-\infty \\ T_r(k)^{-1}(R_r(k)e^{ikx}+e^{-ikx}), & x\rightarrow \infty\end{cases}.
%\end{eqnarray*}
We refer to \cite{DT} for the stationary theory. The
functions $T_r$, $T_l$, and $R_l$, $R_r$ appearing here are exactly
the transmission and reflection coefficients we defined via the
''time- dependent'' theory\footnote{We defined $T_r(k)$, $T_l(k)$,
$R_l(k)$, and $R_r(k)$ for positive $k$; here they are extended to
negative $k$ by e.g. $T_l(k)=\overline{T_l(-k)}$.}.

We need the following properties of the scattering matrix
\cite[Theorem 1]{DT}. Let $\C_+\subset \C$ stand for the open upper
half-plane, i.e. $\C_+=\{ \omega\in \C \mid {\rm Re}(\omega)>0\}$.

\begin{theorem}\label{transmission} Let $V$ be a tunnel potential.
\begin{itemize}
\item[(a)] The scattering matrix $\mathsf{S}(k)$ is unitary for all $k\in \R$, $k\mapsto \mathsf{S}(k)$ is continuous, and we have $T:=T_l=T_r$, $T(-k)=\cc{T(k)}$, $R_r(-k)=\cc{R_r(k)}$, $R_l(-k)=\cc{R_l(k)}$.
\item[(b)] The transmission amplitude $T$ can be extended to the upper half plane $\C_+$ in such a way that $T$ is continuous in $\C_+\cup \R$, analytic in $\C_+$, and satisfies $|T(k)|\leq 1$ for all $k\in \C\cup \R$.
\end{itemize}
\end{theorem}

\begin{remark}\rm The restriction for compactly supported potentials is not really necessary for the
scattering approach; we could just as well use a potential with no bound states, and sufficiently rapid decrease at infinity to
ensure that (a) the Hamiltonian is a well-defined selfadjoint operator, (b) wave operators exist and are
complete, (c) the stationary theory works (with the relations \eqref{asymptotic1} and \eqref{asymptotic2}
understood as ''asymptotically'' valid), and (d) the connection between the ''time-dependent'' and stationary pictures is secured. Specific conditions
for each of these requirements can be found in standard literature (see e.g. \cite{Reed1,Reed2, DT}).
\end{remark}

For the ''time zero approach'', which does not \emph{directly} involve scattering theory, we will use the
expansion of the evolution $e^{-itH}\psi$ in terms of the basic solutions $\phi_i(\cdot,k)$.
Such an expansion is traditionally used in the context of stationary scattering theory; the basic solutions
are called ''improper eigenfunctions'' of $H$. In general, the existence of the expansion is a highly
nontrivial problem, which has a long history (we only mention the old work of Titchmarsh \cite
{Titchmarsh}, as well as some relatively recent papers \cite{Christ1,Christ2,Deift}). We will need the
expansion only for tunnel potentials (compactly supported and bounded), for which it is known to hold,
according to the references just mentioned.
%For such potentials, it can be obtained in a simple way by using Sturm-Liouville theory (see e.g. \cite
%{DSII, Coddington}), and calculating the Green function to show that the associated matrix measure
%\cite[p. 1364]{DSII} is diagonal, and simply reduces to Lebesgue measure.
(As in the case of asymptotic completeness, the problems arise mainly for slowly decaying potentials.)

For $\psi$ belonging to the Schwartz space of rapidly decreasing functions, we define
\begin{equation}
(k,\psi)_i := \frac{1}{\sqrt{2\pi}}\int_\R \phi_i(y,-k)\psi(y) \, dx = \frac{1}{\sqrt{2\pi}}\int_\R \cc{\phi_i(y,k)}\psi(y) \, dx.
\end{equation}
Then
\begin{equation}\label{expansion3}
(e^{-itH}\psi)(x) = \frac{1}{\sqrt{2\pi}}\int_{-\infty}^\infty \frac 12[(k,\psi)_1\phi_1(x,k)+(k,\psi)_2\phi_2(x,k)] e^{-itk^2}\, dk.
\end{equation}
Note that the absence of bound states is reflected in this expansion.

\section{Preparing the ''initial state'' of the particle before the tunnel}

In the introduction we already emphasized the importance of \emph{identifying the initial states} to be
the same in both evolutions. In the first two approaches, the identification is done by means of the
scattering theory; for any given vector state $\psi_0\in L^2(\R,dx)$, with $\hat{\psi_0}\in L^2([0,\infty),dk)$ (i.e.
positive momenta), we find $\psi\in L^2(\R,dx)$ such that $e^{-itH}\psi\sim e^{-itH_0}\psi_0$
asymptotically at $t\rightarrow -\infty$, in the sense that the difference goes strongly to zero at this limit.
This just means $\psi=\Omega_-\psi_0$. Note that here $\psi_0$ is not the initial state, because the
''initial time'' is considered to be $t=-\infty$.
According to a well-known result called ''scattering into cones'' \cite{Dollard}, this setup means that the
particle is initially localized ''far to the left'' of the potential at $t\rightarrow -\infty$, and is ''going to the
right'' at any time $t$.

Obviously, the pure state $\psi_0$ can also be replaced by a general state operator $\rho_0$. Then the
condition of positive momenta is $\tr[\rho_0\check{P}_+]=1$, or, equivalently,
$\check{P}_+\rho_0\check{P}_+=\rho_0$.

In the ''time zero approach'', we take an interval $(x_0,x_1)$ which includes the support of the tunnel
potential $V$, and at the initial time $t=0$ we prepare a state $\psi_0\in L^2((-\infty, x_0],dx)$. Then for
$t>0$ the tunneled and freely evolved states are simply $e^{-itH}\psi_0$ and $e^{-itH_0}\psi_0$, respectively.
We let $P_{{\rm init}}$ denote the projection onto $L^2((-\infty, x_0],dx)$, so that we can state the initial
condition for a general state $\rho_0$ as $\tr[\rho_0 P_{{\rm init}}]=1$.

\section{The ''detection'' of the particle after the tunnel}

We describe here in detail the detection method in each of the three schemes; in each case, we end up
with two relevant probabilities, corresponding to the tunnel particle and the free particle, respectively.

\subsection{Approach I: arrival time}

For an arrival time observable $E$, the number $\tr[E((-\infty, t])\rho]$ is interpreted as the probability that
a particle whose state is $\rho$ at $t=0$ arrives at a certain point (which depends on $E$) during the
time $(-\infty, t]$. As explained in the introduction, the idea is to compare the arrival time probability of a
particle moving in the presence of a potential, with the corresponding probability of a freely evolving
particle, with initial states identified as above.

For a free particle, an arrival time observable $E_0:\h B(\R)\to \h L(\hi)$ must satisfy the covariance
condition \eqref{timeobs}, and the additional condition that $E(\R)=\check{P}_+$. As explained before,
this means that the observable is only sensitive to positive momenta; particles traveling ''to the left'' will
not be detected. The correct time observable $E$ for the evolution according to $H$ should satisfy
\eqref{timeobs} with $H_0$ replaced by $H$, because $H$ generates the time translations for this
system.

With $E_0$ and $E$ chosen this way, and given a pure state $\psi_0$
as in the preceding section, with $\psi=\Omega_-\psi_0$, the arrival
probabilities to be compared are of the form $\langle \psi
|E((-\infty,t])\psi \rangle$ and $\langle \psi_0
|E_0((-\infty,t])\psi_0\rangle$. In order to ensure that the
comparison is meaningful, the observables $E$ and $E_0$ have to be
''the same after the scattering event'', i.e., at large times
$t\rightarrow\infty$. Accordingly, we require that for any given
$\vp_0\in \hi$,
$$
\langle \vp |E(B) \vp\rangle = \langle \vp_0 |E_0(B)\vp_0\rangle, \ \text{ for all } B\in \h B(\R),
$$
with $\vp = \Omega_+\vp_0$. This just amounts to saying that arrival time probabilities corresponding to
the two evolutions should coincide for states which will become asymptotically equal at
$t\rightarrow\infty$. This condition is equivalent to the requirement $E(B) = \Omega_+E_0(B)\Omega_+^*$.
Note that for any arrival time observable $E_0$ corresponding to $H_0$, the observable $B\mapsto E(B):=\Omega_+E_0(B)\Omega_+^*$ indeed satisfies \eqref{timeobs} with $H_0$ replaced by $H$ because of \eqref{intertwining}.

Hence, in the end we actually need only the observable $E_0$, which satisfies \eqref{timeobs}; we
compare
$$\langle \psi |E((-\infty,t])\psi\rangle = \langle S\psi_0 |E_0((-\infty,t])S\psi_0\rangle$$ with
$$\langle \psi_0 |E_0((-\infty,t])\psi_0\rangle.$$
Moreover, since $E_0((-\infty, t])\leq \check{P}_+$,
and $\psi_0\in \check{P}_+L^2(\R,dx)$, we can simply replace $S$ by the transmission amplitude
$\check{P}_+S\check{P}_+=T(P)$ in the former. This just means that because the observable is
only sensitive to positive momenta, it does not ''see'' the reflected part of the state. Using an arbitrary
state $\rho_0$ with $\tr[\rho_0\check{P}_+]=1$, we thus compare
\begin{align}
\mathbf{p}_{\rm I}(t) &:= \tr[T(P)^*E_0((-\infty,t])T(P)\rho_0],\nonumber\\
\mathbf{p}^0_{\rm I}(t) &:= \tr[E_0((-\infty,t])\rho_0],\label{arrivalprobs}
\end{align}
where the index ${\rm I}$ refers to this first approach.

\subsection{Approach II: localization measurement}

Here we do not have the problem of identifying the observables; we make the \emph{same} localization
measurement $E$ for both tunnel and particle case, at a large preset time $t$. At this time, the
corresponding states are $e^{-itH}\psi$ and
$e^{-itH_0}\psi_0$. Here ''large'' time means that we are in the asymptotic regime, i.e. we identify
$e^{-itH}\psi\sim e^{-itH_0}S\psi_0$ at $t\rightarrow\infty$, in the sense that the difference goes strongly
to zero.

Now the reflected part $\check{P}_-S\psi_0$ does not contribute to the localization measurement since it
''moves to the left'' while we localize in $[a,\infty)$. In the case of sharp localization (corresponding to the
spectral measure $E^Q$ of $Q$) this is a well-known consequence of the ''scattering into cones'' -
theorem, and can be derived from the asymptotic form for the free propagator (see e.g. \cite[p. 60]{Reed1}):
$$
s-\lim_{t\rightarrow \infty} E^Q([a,\infty))e^{-itH_0}\check{P}_- =0.
$$
%$$\lim_{t\rightarrow\infty} \|\chi_{[a,\infty)}(M)(e^{-itH_0}\vp- C_t\vp)\|=0,$$
%for any $\vp\in L^2(\R)$, where
%$(C_t\vp)(x) = (2it)^{-\frac 12} e^{ix^2/(4t)}\hat{\vp}(x/(2t))$.
%If $\vp\in \check{P}_-(\hi)$, i.e. $\int_0^\infty |\hat{\vp}(k)|^2\, dk = 0$, then
%$$\|\chi_{[a,\infty)}(M)C_t\vp\|^2 = \int_a^\infty (2t)^{-1} |\hat{\vp}{k/(2t)}|^2 dk
%=\int_{a/(2t)}^\infty |\hat{\vp}(k)|^2\, dk \rightarrow  \int_0^\infty |\hat{\vp}(k)|^2\, dk=0.
%$$
%as $t\rightarrow\infty$.
As we mentioned when introducing the localization observables, each of them is a convolution of the
sharp localization $E^Q$ with a probability measure. Using this fact, the same limit result is easily
proved also for the general case:
\begin{lemma}\label{asymptoticloc}
Let $E:\h B(\R)\to \h L(L^2(\R,dx))$ be an arbitrary localization observable, and $a\in \R$. Then
$$
s-\lim_{t\rightarrow\infty} E([a,\infty))e^{-itH_0}\check{P}_- = 0.
$$
\end{lemma}
\begin{proof} Suppose that $\vp\in \check{P}_-L^2(\R,dx)$, $\|\vp\|=1$, and put $\vp_t = e^{-itH_0}\vp$ for $t\in \R$.
Let $\mu:\h B(\R)\to [0,1]$ be a probability measure, such that $E= \mu*E^Q$. For any $\psi\in \hi$ let
$E_{\psi,\vp_t}$ denote
the complex measure $B\mapsto \langle \psi |E(B)\vp_t\rangle$, and define $E^Q_{\psi,\vp_t}$
similarly. Then we have $E_{\psi,\vp_t}= \mu * E^Q_{\psi,\vp_t}$, so that $E_{\psi,\vp_t}([a,\infty)) = (\mu\times E^Q_{\psi,\vp_t})(A)$,
where $A=\{ (x,y)\in \R^2\mid x+y\geq a\}$. Let $\epsilon >0$. Since $\mu$ is a finite measure, there
exists an $b\in (0,\infty)$ with $\mu\big((-\infty, -b)\cup (b,\infty)\big)<\frac \epsilon 2$. Let $t_0>0$
be such that for $t\geq t_0$ we have $\|E^Q([a-b,\infty))\vp_t\|<\frac {\epsilon}{ 2}$.
Then for $t\geq t_0$, and any unit vector $\psi\in \hi$, we have
% For any Borel set $X\subset [a-b,\infty)$, we have
%$$|\langle \psi |E^Q(X)\vp_t\rangle|\leq \sqrt{\langle \psi |E^Q(X)\psi\rangle\langle \vp_t |E^Q(X)\vp_t\rangle} \leq \sqrt{\langle \vp_t |E^Q([a-b,\infty)\vp_t\rangle}=\|E^Q([a-b,\infty))\vp_t\|< \frac \epsilon 2,$$
%so
$| E^Q_{\psi,\vp_t}|([a-b,\infty))< \frac \epsilon 2$ (where the $|\cdot |$ stands for the total
variation of the measure), and consequently,
\begin{eqnarray*}
|E_{\psi,\vp_t}([a,\infty))| &\leq& (\mu\times |E^Q_{\psi,\vp_t}|)(A)\\
&\leq& \mu\big((-\infty,-b)\cup (b,\infty)\big)|E^Q_{\psi,\vp_t}|(\R)+\mu(\R)|E^Q_{\psi,\vp_t}|([a-b,\infty))<\epsilon.
\end{eqnarray*}
This implies $\|E([a,\infty))\vp_t\|<\epsilon$, so the proof is complete.
\end{proof}
Thus, we may identify
$$E([a,\infty))e^{-itH}\psi\sim E([a,\infty))e^{-itH_0}S\psi_0\sim E([a,\infty))e^{-itH_0}T(P)\psi_0,$$
in the sense of strong asymptotic convergence.
%This also implies that in the weak (and hence ultraweak) operator topology,
%$e^{itH_0}E([a,\infty))e^{-itH_0}P_-\rightarrow 0$, and $P_-e^{itH_0}E([a,\infty))e^{-itH_0}\rightarrow 0$,
% as $t\rightarrow \infty$. In view of the decomposition \eqref{sdecomp}, this means that for any given
% state $\rho$, we may approximate
% $$
% \tr[E([a,\infty))e^{-itH_0}S\rho S^*e^{itH_0}]\sim \tr[E([a,\infty))e^{-itH_0} \check{P}_+S\rho S^*\check{P}_+e^{itH_0}]
% $$
%arbitrarily well for large enough values of $t$.
Hence, in this second approach, we want to compare the localization probabilities
\begin{align}
\mathbf{p}_{\rm II}(a) &:= \tr[T(P)^*E([a,\infty))T(P)e^{-itH_0} \rho_0 e^{itH_0}];\nonumber\\
\mathbf{p}^0_{\rm II}(a) &:= \tr[E([a,\infty))e^{-itH_0} \rho_0 e^{itH_0}],\label{localizationprobs}
\end{align}
for arbitrary states $\rho_0$ with $\tr[\rho_0\check{P}_+]=1$.

\subsection{Approach III: time zero}

Here we make a sharp localization measurement, with the observable $E^Q$, at time $t>0$
corresponding to an interval $[a,\infty)$, where $a\geq x_1$. Recall that $t=0$ is the initial time, the
support of the potential is contained in $(x_0,x_1)$, and the initial state $\rho_0$ satisfies
$\tr[\rho_0P_{\rm init}]=1$, where $P_{\rm init}=E^Q((-\infty,x_0])$. In view of the result
 \eqref{evolutionoperator} below, it will be convenient to put $P_{\rm final}:=E^Q([x_1,\infty))$.
 %this is the projection onto right hand side of the potential, where the measurement is made.

If $\rho_0 = |\psi_0\rangle \langle \psi_0 |$, where $\psi_0$ belongs to the Schwartz space,
we can use the expansion \eqref{expansion3}. Since the support of $\psi_0$ is contained in
$(-\infty, x_0]$, we can express $(\psi_0,k)_i$ in terms of the Fourier transform; indeed, by
\eqref{asymptotic1} and \eqref{asymptotic2}, as well as Theorem \ref{transmission} (a), we get
\begin{align*}
(\psi_0,k)_1 &= \frac{1}{\sqrt{2\pi}}\int_{-\infty}^{x_0}\psi_0(x)\phi_1(x,-k)\, dx = \hat{\psi_0}(k) +R_l(-k) \hat{\psi_0}(-k),\\
(\psi_0,k)_2 &= \frac{1}{\sqrt{2\pi}}\int_{-\infty}^{x_0}\psi_0(x)\phi_2(x,-k)\, dx = T(-k)\hat{\psi_0}(-k).
\end{align*}
The expansion \eqref{expansion3} now takes the form
$$
(e^{-itH}\psi_0)(x) = \frac{1}{\sqrt{2\pi}}\int_{-\infty}^\infty \frac 12 [(\hat{\psi_0}(k) +R_l(-k) \hat{\psi_0}(-k))\phi_1(x,k)+T(-k)\hat{\psi_0}(-k)\phi_2(x,k)]e^{-itk^2}\, dk.
$$
Using again the relations \eqref{asymptotic1} and \eqref{asymptotic2} we get, for $x\geq x_1$,
\begin{eqnarray*}
(e^{-itH}\psi_0)(x) %&=& \frac{1}{\sqrt{2\pi}}\int_{-\infty}^\infty \frac 12 \big[(\hat{\psi_0}(k) +R_l(-k) \hat{\psi_0}(-k))T(k)e^{ikx}+ T(-k)\hat{\psi_0}(-k)[R_r(k)e^{ikx}+e^{-ikx}]\big]e^{-itk^2}\, dk\\
&=& \frac{1}{\sqrt{2\pi}}\int_{-\infty}^\infty \frac 12 \big[\hat{\psi_0}(k)T(k)e^{ikx}+ \hat{\psi_0}(-k)T(-k)e^{-ikx}\\
&+& [R_l(-k)T(k)+T(-k)R_r(k)]\hat{\psi_0}(-k)e^{ikx}\big]e^{-itk^2}\, dk\\
&=& \frac{1}{\sqrt{2\pi}}\int_{-\infty}^\infty \frac 12 \big[\hat{\psi_0}(k)T(k)e^{ikx}+ \hat{\psi_0}(-k)T(-k)e^{-ikx}\big]e^{-itk^2}\, dk,\\
%&=& \frac{1}{\sqrt{2\pi}}\int_{-\infty}^\infty \hat{\psi_0}(k)T(k)e^{ikx}e^{-itk^2}\, dk,
\end{eqnarray*}
where the second equality is obtained by Theorem \ref{transmission} (a). Hence, in this case, the
evolution is simply given by
\begin{equation}\label{finalexpansion}
(e^{-itH}\psi_0)(x)=\frac{1}{\sqrt{2\pi}}\int_{-\infty}^\infty e^{ikx}T(k)\hat{\psi_0}(k)e^{-itk^2}\, dk, \,\,\, x\geq x_1,
\end{equation}
or, equivalently,
$$
P_{\rm final}e^{-itH}\psi_0 = P_{\rm final}F^{*}T(\hat{P})e^{-it\hat{P}^2}F\psi_0.
$$
Since the operators $F^{*}$, $T(\hat{P})$, $e^{-it\hat{P}^2}$, and $F$ are all bounded, and the Schwartz space
functions with support contained in $(-\infty, x_0]$ are dense in $L^2((-\infty, x_0])$, the above relation is
equivalent to the operator equality $P_{\rm final} e^{-itH}P_{\rm init} = P_{\rm final}F^{*}T(\hat{P})e^{-itM^2}FP_{\rm init}$, i.e.
\begin{equation}\label{evolutionoperator}
P_{\rm final}e^{-itH}P_{\rm init}=P_{\rm final}T(P)e^{-itH_0}P_{\rm init}.
\end{equation}
We want to compare the sharp localization probabilities
$\mathbf{p}_{\rm III}(a):=\tr[ E^Q([a,\infty))e^{-it H}\rho_0 e^{itH}]$ and
$\mathbf{p}^0_{\rm III}(a):=\tr[E^Q([a,\infty))e^{-it H_0}\rho_0 e^{itH_0}]$, where $\tr[\rho_0 P_{\rm init}]=1$,
and $t>0$, $a\geq x_1$ are fixed. Using \eqref{evolutionoperator}, and noting that
$P_{\rm init}\rho_0 P_{\rm init} = \rho_0$, $P_{\rm final}E^Q([a,\infty))P_{\rm final}= E^Q([a,\infty))$,
we immediately get
\begin{align}
\mathbf{p}_{\rm III}(a) &= \tr [T(P)^*E^Q([a,\infty))T(P)e^{-itH_0}\rho_0 e^{itH_0}];\nonumber\\
\mathbf{p}_{\rm III}^0(a) &=\tr[E^Q([a,\infty))e^{-it H_0}\rho_0 e^{itH_0}].\label{timezeroprobs}
\end{align}
Note that although this expression is similar to the one in the above scattering theory with localization -
approach, here $\rho_0$ has also negative momentum components, and thus also the values of $T$ for
negative argument are used.

\section{A mathematical theorem}

By looking at the probabilities \eqref{arrivalprobs}, \eqref{localizationprobs}, and \eqref{timezeroprobs} it
is clear that the problem of comparison is similar in each case. As we have already pointed out, we only
need to transform the covariant observable in each case into the spectral representation of the
associated generator, which is $H_0=P^2$ in the arrival time case, and $P$ in the localization case.
Since the spectrum of $H_0$ is $[0,\infty)$, while the spectrum of $P$ is $(-\infty, \infty)$, we need
to consider both $L^2([0,\infty))$ and $L^2(\R)$. In their respective spectral representations, the operators $H_0$ and $P$ act as the multiplication operator $M$.

In order to deal with both cases, we formulate the essential result (Theorem \ref{theorem2} below) for a subset $I\in \h B(\R)$ (either $\R$ or $[0,\infty)$, in practise), and the corresponding subspace $L^2(I)=\chi_I(M)L^2(\R)$. For a
Borel function $f:\R\to \C$, this subspace  is obviously invariant under the multiplication operator $f(M)$
for any $f:\R\to \C$, this operator being the multiplication by the restriction of $f$ to $I$. For notational
simplicity, we will use the symbol $f(M)$ also to denote this restricted operator. Such a restriction will be
needed for the transmission amplitude, when we restrict it to positive momenta.

It should be clear from the above discussion that here we do not fix the physical interpretation of $M$; consequently, also the operators $\hat{P}_\pm$ will appear in the following lemma without such interpretation. Since these projections are not involved in the statement of Theorem \ref{theorem2} (which is the only thing we need to refer to) they can just be regarded as mathematical auxiliaries in this section.

The following lemma contains the essential mathematical ingredient we need.
The symbol $\C_+\subset \C$ stands for the upper open half-plane., i.e. $\C_+=\{ \omega\in \C \mid {\rm Re}(\omega)>0\}$.

\begin{lemma}\label{theorem1}
Let $g:\R\to \C$ be a measurable function, and suppose that $\tilde{g}:\C_+\cup \R\to \C$ is an extension
of $g$ such that
\begin{itemize}
\item[(i)] $\tilde{g}$ is analytic on $\C_+$;
\item[(ii)] $|\tilde{g}(\omega)|\leq 1$ for all $\omega \in \C_+\cup \R$;
\item[(iii)] $\lim_{b\rightarrow 0+} \tilde{g}(a+ib) = g(a)$ for almost all $a\in \R$.
\end{itemize}
Then $$g(M)^*\hat{P}_+g(M)\leq \hat{P}_+.$$
\end{lemma}
\begin{proof} Put $\hi=L^2(\R)$ for convenience. We first note that it is sufficient to prove
\begin{equation}\label{invariance}
g(M)\hat{P}_-\hi\subset \hat{P}_-\hi,
\end{equation}
Indeed, \eqref{invariance} implies $\langle \vp |g(M)^*\hat{P}_+g(M)\psi\rangle = \langle g(M)\vp |\hat{P}_+g(M)\psi\rangle = 0$ for
$\vp\in \hat{P}_+(\hi)^\perp=\hat{P}_-(\hi)$, $\psi\in \hi$, that is, $g(M)^*\hat{P}_+g(M)\hi\subset \hat{P}_+\hi$;
this gives $g(M)^*\hat{P}_+g(M) \leq \hat{P}_+$, because $g(M)^*\hat{P}_+g(M)$ is a positive
operator with norm less than one by (ii).

Now $\vp\in \hat{P}_-(\hi)$ if and only if $F^*\vp$ is supported in $(-\infty,0]$, or equivalently, if and only if $F\vp$ is supported in $[0,\infty)$. To prove
\eqref{invariance}, let $\vp\in \hat{P}_-(\hi)$, and define $\tilde{\vp}:\C_+\to \C$ via
$$
\tilde{\vp}(\omega):= \frac{1}{\sqrt{2\pi}}\int_\R e^{i\omega y} (F\vp)(y)\, dy, \ \ \omega \in \C_+,
$$
where the integral exists because $y\mapsto | e^{i\omega y}| = e^{-{\rm Im}(\omega) y}$ is now square
integrable over $[0,\infty)$. Then $\tilde{\vp}\in H^+$, where
%This exponential decrease also ensures that $\tilde{\vp}$ is analytic on the domain $\C_+$. Since clearly $\tilde{\vp}(a+ib)= F^{-1}[e^{-b \cdot} F\vp](a)$, it follows that
%$$\int_\R |\tilde{\vp}(a+ib)|^2\, da = \int e^{-2by}|F\vp(y)|^2\, dx \leq \|\vp|^2, \ \ b>0,$$
%so $\tilde{\vp}\in H^+$, where
$$
H^+ =\left\{ h:\C_+ \to \C\mid h \text{ analytic}, \ \sup_{b>0} \int_\R |h(a+ib)|^2\, da <\infty\right\}
$$
is called the Hardy class; see e.g. \cite[p. 161-162]{Dym}. According to this reference, the elements
$h\in H^+$ are characterized as precisely those functions $h:\C_+\to \C$ for which there exists a function
$f\in L^2(\R)$ supported on $[0,\infty)$, such that
$$
h(\omega) = \frac{1}{\sqrt{2\pi}}\int_\R e^{i\omega y} f(y)\, dy, \ \ \omega\in \C_+ .%\, = (F^{-1}[e^{-{\rm Im}(\omega)\cdot }f])({\rm Re}(\omega)).
$$
Here $f$ is recovered via the $L^2(\R)$-limit $F^{*} f= \lim_{b\rightarrow 0} h_b$, where $h_b(a) := h(a+ib)$.
%This last fact follows immediately from the dominated convergence theorem, since
%$$
%\int_\R |h_b(a)-(F^{-1}f)(a)|^2\, da = \int_\R |F^{-1}[e^{-b\cdot}f](a)-(F^{-1}f)(a)|^2\, da = \int_\R (e^{-bx}-1)^2|f(a)|^2\, da.
%$$
Since $\tilde{\vp}\in H^+$, the assumptions (i) and (ii) clearly imply that also
$\omega\mapsto \tilde{g}(\omega)\tilde{\vp}(\omega)$ is an element of $H^+$. Hence, there exists a
$\psi\in \hat{P}_-(\hi)$ with
$$
\tilde{g}(\omega)\tilde{\vp}(\omega) = \frac{1}{\sqrt{2\pi}} \int_\R e^{i\omega y} (F\psi)(y)\, dy, \ \ \omega\in \C_+.
$$
To conclude the proof of \eqref{invariance}, we have to show that $\psi=g(M)\vp$. According to the
definition of $H^+$, we have $\psi_b, \vp_b\in L^2(\R)$ for any $b>0$, where $\psi_b(a):=\tilde{g}(a+ib)\tilde{\vp}(a+ib)$ and $\vp_b(a):=\tilde{\vp}(a+ib)$. Let $g_b$ be the function $a\mapsto \tilde{g}(a+ib)$
for each $b> 0$, so that $\psi_b = g_b(M)\vp_b$, for $b>0$. As mentioned above, we then have the
$L^2(\R)$-limits $\lim_{b\rightarrow 0} \psi_b = \psi$ and $\lim_{b\rightarrow 0} \vp_b = \vp$. Since the
family of multiplication operators $\{ g_b(M)\mid b>0\}$ is uniformly bounded because of (ii), and
$g_b(M)$ tends strongly to $g(M)$ as $b\rightarrow 0$ (by (iii), (ii) and the dominated convergence),
it follows that $\psi = g(M)\vp$. The proof is complete.
\end{proof}

Next we need information on the structure of covariant observables. This is well-known (see e.g.
\cite{Holevo2,Werner}). For the purposes of this paper, it is convenient to state it in the following form,
which can be immediately specialized from the general construction procedure, obtained by combining
Mackey's imprimitivity theorem with a dilation argument \cite{Werner}.

\begin{lemma} Let $I\subset \R$ be a Borel set, and make the usual identification $L^2(I)=\chi_I(M)L^2(\R) \subset L^2(\R)$. Let $E:\h B(\R)\to \h L(L^2(I))$ be a positive operator measure with
the covariance property
$$
e^{-iaM} E(B)e^{iaM} = E(B+a), \ \ \text{for all } B\in \h B(\R),\, a\in \R.
$$
Then $E$ is of the form
$$E(B) = \Psi(F\chi_B(M)F^*), \  \ \text{for all } B\in \h B(\R),$$
where $\Psi:\h L(L^2(\R))\to \h L(L^2(I))$ is a positive linear map satisfying
$$\Psi\big(f(M)^*Af(M)\big) = f(M)^*\Psi\big(A\big)f(M),$$
for any bounded Borel function $f:\R\to \C$ and $A\in \h L(L^2(\R))$.
\end{lemma}

With this result, the ingredient given by Lemma \ref{theorem1} can be
imported into the covariant setting:

\begin{theorem}\label{theorem2} Let $I\in \h B(\R)$,
and let $E:\h B(\R)\to \h L(L^2(I))$ be a POM with the covariance property
$$
e^{-iaM} E(B)e^{iaM} = E(B+a), \ \ B\in \h B(\R), \, a\in \R.
$$
Let $g:\R\to \C$ satisfy the conditions of Lemma \ref{theorem1}. Then
$$
g(M)^*E([a,\infty))g(M) \leq E([a,\infty)), \ \ a\in \R.
$$
\end{theorem}
\begin{proof}
First note that by covariance, $E([a,\infty))$ is unitarily equivalent with $E([0,\infty))$ via the unitary $e^{iaM}$,
which commutes with $g(M)$. Hence, it suffices to prove the inequality for $t=0$.
%Indeed, once this is established, we get
%\begin{eqnarray*}
%T(M)^*E((-\infty,t])T(M) &=& T(M)^*e^{itM}E((-\infty,0])e^{-itM}T(M)\\
%& =& e^{itM}T(M)^*E((-\infty,0])T(M)e^{-itM}\\
% &\leq& e^{itM}E((-\infty,0])e^{-itM} = E((-\infty,t])
%\end{eqnarray*}
%by covariance and the fact that $T(M)$ commutes with $M$.
To this end, let $\Psi:\h L(L^2(\R)))\to L(L^2(I))$ be the map given by the above lemma, corresponding to $E$.
Since $g$ is a bounded Borel function, and $\Psi$ is positive and linear, Lemma \ref{theorem1} gives
\begin{eqnarray*}
g(M)^*E([0,\infty))g(M) &=& g(M)^*\Psi\big(FP_+F^*\big)g(M)= \Psi\big(g(M)^*\hat{P}_+g(M)\big)\\
&\leq& \Psi\big(\hat{P}_+\big) = E([0,\infty)),
\end{eqnarray*}
and the proof is finished.
\end{proof}

\begin{remark}\rm Notice that although the range of all the $E(B)$ are here asserted to be in $L^2(I)$,
and, consequently, only the restriction $g|_I$ appears in the inequality of the above theorem, it is
still necessary to have $g$ as a function on the whole $\R$. This is because otherwise we could not
move $g(M)$ inside the argument of $\Psi$, and apply Lemma \ref{theorem1}.
\end{remark}

\section{Results}

We can now apply Theorem \ref{theorem2} of the preceding section to the three relevant cases,
in order to compare the probabilities in \eqref{arrivalprobs}, \eqref{localizationprobs}, and \eqref{timezeroprobs}.

\subsection{Approach I: arrival time}
We pass from ''position representation'' to the ''energy representation'' by means of the unitary operator $U_+F$, where the unitary
$U_+:L^2([0,\infty),dk)\to L^2([0,\infty), dE)$ is given by
$$(U_+\vp)(E) = 2^{-\frac 12} E^{-\frac 14} \vp(\sqrt{E}), \, \, \, \text{ for all } \vp\in L^2([0,\infty), dk).$$
The transmission amplitude $T(\hat{P})=T(M)$ (which already acts in $L^2([0,\infty),dk)$), is transformed into $U_+T(M)U_+^*=T'(M)$, where $T':\R\to \C$ is given as $T'(E) = T(\sqrt{E})$, and the arrival time observable $E_0:\h B(\R)\to \h L(L^2(\R,dx))$
transforms into the observable $E'_0:\h B(\R)\to \h L(L^2([0,\infty),dE))$, where
$$E'_0(B) = U_+FE_0(B)F^*U_+^*.$$
Note that since the range of $E_0(B)$ is included in $\check{P}_+L^2(\R,dx)$ (positive momenta),
the range of $FE_0(B)$ is included in $L^2([0,\infty),dk)$, so that this is indeed well-defined.
The probabilities in \eqref{arrivalprobs} now have the form
\begin{align*}
\mathbf{p}_{\rm I}(t) &= \tr[T'(M)^*E_0'((-\infty,t])T'(M)U_+\hat{\rho_0}U_+^*],\\
\mathbf{p}^0_{\rm I}(t) &= \tr[E_0'((-\infty,t])U_+\hat{\rho_0}U_+^*].
\end{align*}
Choosing the square root branch
$$\{ re^{i\theta} \mid r>0, \, -\pi <\theta\leq \pi \} \ni re^{i\theta}\mapsto \sqrt{r}e^{\tfrac 12 \theta}
\in \{ re^{i\theta} \mid r>0, \, -\pi/2 <\theta\leq \pi/2 \},$$
and using Theorem \ref{transmission} (b), we see that $T' = T(\sqrt{\cdot})$ satisfies the conditions
of Lemma \ref{theorem1}. Due to the intertwining relations $U_+e^{itM^2}=e^{itM}U_+$, $Fe^{itH_0} = e^{itM^2}F$ (which hold for all $t\in \R$), the observable $E_0'$ satisfies the covariance condition
$$
e^{itM}E_0'(B)e^{-itM} = E_0'(B+t), \ \ t\in \R,
$$
which implies that $B\mapsto E_0'(-B)$ (and not $E_0'$ itself) satisfies the covariance condition of
Theorem \ref{theorem2}. This reflection simply means that the conclusion of Theorem \ref{theorem2} holds for observable $E_0'$ with the intervals $-[-t,\infty)=(-\infty, t]$, and we get
\begin{equation}\label{arrivalresult}
\mathbf{p}_{\rm I}(t)\leq \mathbf{p}^0_{\rm I}(t), \ t\in \R,
\end{equation}
for any state $\rho_0$ with $\tr[\rho_0\check{P}_+]=1$. This means that \emph{the probability of arrival
by the time $t$ is never larger for the tunneled particle.}

\subsection{Approach II: localization measurement}
Here the application of Theorem \ref{theorem2} is more straightforward: In comparing the probabilities \eqref{localizationprobs}, we only need to pass to the momentum
representation; define $\hat{E}:\h B(\R)\to \h L(L^2(\R,dk))$ as $\hat{E}(B)= FE(B)F^*$. Because of the
intertwining $Fe^{-itP} = e^{-itM}F$, this observable now satisfies the covariance condition
of Theorem \ref{theorem2}. In addition, $T$ satisfies the conditions of Lemma \ref{theorem1} by Theorem \ref{transmission} (b), so Theorem \ref{theorem2} immediately gives
\begin{equation}\label{localizationoperatorinequality}
T(M)^*\hat{E}([a,\infty))T(M) \leq \hat{E}([a,\infty)),  \, \, \text{ for all } a\in \R.
\end{equation}
Applied to the probabilities \eqref{localizationprobs}, this implies
$$
\mathbf{p}_{\rm II}(a)\leq \mathbf{p}^0_{\rm II}(a), \  a\in \R,
$$
for any state $\rho_0$ with $\tr[\rho_0\check{P}_+]=1$. (Recall also that the localization measurement is made at large time $t$; for small $t$, the comparison is just between the transmitted parts.) The result means that \emph{the probability of having passed the point $a$ at time $t$ (in the sense of the particular localization used) is never larger for the tunneled particle.}

\subsection{Approach III: time zero}
Finally, consider the probabilities \eqref{timezeroprobs}. Since $E^Q$ is a localization observable, the
inequality \eqref{localizationoperatorinequality} immediately applies also to this case, and we get
$$
\mathbf{p}_{\rm III}(a)\leq \mathbf{p}^0_{\rm III}(a), \ a\geq x_1,
$$
for any state $\rho_0$ with $\tr[\rho_0P_{\rm init}]=1$. This has the same meaning as in the above case,
except that here the localization is only understood in the sharp sense, and the time $t$ at which the measurement is performed can be any $t>0$.

\

\noindent {\bf Acknowledgment.} One of the authors (J. K.) was supported by Finnish Cultural Foundation during the preparation of the manuscript.

\end{document}